\documentclass[a4paper,10pt]{article}

\usepackage[colorinlistoftodos]{todonotes}

\usepackage[left=0.93in,right=0.93in,top=1in,bottom=1in]{geometry}

\usepackage[utf8]{inputenc}
\usepackage[english]{babel}
\usepackage[babel]{csquotes}
\usepackage{amsthm, amsmath, amssymb,amsfonts,latexsym,soul,enumerate}
\usepackage{mathtools} %for \xRightarrow
\usepackage{enumitem} 
\usepackage{bbm}       %for indicatrix function \mathbbm{1}
\usepackage{mathrsfs}  %for mathscr

    \usepackage[numbers,sort]{natbib}

\usepackage[colorlinks]{hyperref}
\usepackage{chemfig,siunitx}
\setcompoundsep{7em} % (not exactly) the length of the arrows

\newcommand{\Sp}{\mathcal{X}}
\newcommand{\C}{\mathcal{C}}
\newcommand{\R}{\mathcal{R}}
\newcommand{\G}{\mathcal{G}}

\newcommand{\RR}{\mathbb{R}}
\newcommand{\ZZ}{\mathbb{Z}}

\DeclareMathOperator{\spann}{span}

\binoppenalty=\maxdimen
\relpenalty=\maxdimen

\newtheorem{theorem}{Theorem}[section]

\newtheorem{proposition}[theorem]{Proposition}

\theoremstyle{remark}
\newtheorem{remark}{Remark}[section]

\theoremstyle{definition}
\newtheorem{definition}{Definition}[section]

\setlist[description]{wide, style=sameline}

\begin{document}

\title{Discrepancies between extinction events and boundary equilibria in reaction networks}

%\title{Relationship between extinction events of stochastic models and equilibria of deterministic models of reaction networks, under mass action kinetics}

\author{David F.\ Anderson\footnotemark[1] \and Daniele Cappelletti\footnotemark[2]}

\footnotetext[1]{Department of Mathematics, University of Wisconsin-Madison, email: \texttt{anderson@math.wisc.edu}. Supported by  Army Research Office grant W911NF-18-1-0324.}
\footnotetext[2]{Department of Mathematics, University of Wisconsin-Madison, email: \texttt{cappelletti@math.wisc.edu}.}

 \tikzset{every node/.style={auto}}
 \tikzset{every state/.style={rectangle, minimum size=0pt, draw=none, font=\normalsize}}
  \tikzset{bend angle=15}

\maketitle
 
\begin{abstract}
Reaction networks are mathematical models of interacting chemical species that are primarily used in biochemistry. There are two modeling regimes that are typically used, one of which is deterministic  and one that is stochastic.  In particular, the deterministic model consists of an autonomous system of differential equations, whereas the stochastic system is a continuous time Markov chain.
Connections between the two modeling regimes have been studied since the seminal paper by Kurtz (1972), where the deterministic model is shown to be a limit of a properly rescaled stochastic model over compact time intervals. Further, more recent studies have connected the long-term behaviors of the two models when the reaction network satisfies certain graphical properties, such as weak reversibility and a deficiency of zero.

These connections have led some to conjecture a link between the long-term behavior of the two models exists, in some sense.
%These connections have led some to believe a number of things pertaining to the possible behavior of the two modeling choices. 
 In particular, one is tempted to believe  that positive recurrence of all states for the stochastic model implies the existence of positive equilibria in the deterministic setting, and that boundary equilibria of the deterministic model imply the occurrence of an extinction event in the stochastic setting. We prove in this paper that these implications do not hold in general, even if restricting the analysis to networks that are bimolecular and that conserve the total mass. In particular, we disprove the implications in the special case of models that have absolute concentration robustness, thus answering in the negative a  conjecture stated in   the literature  in 2014.

\end{abstract}

\section{Introduction}

Reaction systems are mathematical models that are used to describe the dynamical behavior of interacting chemical species.  Such models are often utilized in the biochemical setting, where they describe   biological processes.    Traditionally, we distinguish between a deterministic and a stochastic modeling regime, with the deterministic regime appropriate when the counts are so high that the concentrations of the species can be well modeled via a set of autonomous differential equations and with the stochastic model appropriate when the counts are low.  For the stochastic model, one usually assumes the counts of the different chemical species involved  evolve according to a continuous time Markov chain in $\ZZ_{\ge 0}^d$ (where $d$ is the number of distinct chemical species).

 It is natural to wonder about  the relationship between the stochastic and  deterministic models for reaction systems. The first paper in this direction was \cite{kurtz:classical}, where it was shown that on compact time intervals, the deterministic model is the weak limit of the stochastic model, conveniently rescaled, when the initial counts of molecules go to infinity in an appropriate manner (see Theorem \ref{thm:classical_scaling} below). Followup works consider piecewise deterministic limits in multiscale settings \cite{KK:multiscale, PP:multiscale, KR:multiscale, KR:enzyme}, and connections have been found between equilibria of the deterministic model and stationary distributions of the stochastic model, under certain assumptions \cite{ACK:poisson, AC2016,CW:poisson, CJ:graphical}. Further connections have been studied in terms of Lyapunov functions of the deterministic model and stationary distributions of the stochastic model \cite{ACGW:lyapunov}.

The focus of this paper is on the relationship (or lack thereof) between the occurrence of extinction events in the stochastic model and the equilibria of the corresponding deterministic model. The relevance of this question resides in the fact that equilibria of the deterministic model are typically easier to analyze than the state space of the stochastic model, so finding a link between the two is desirable. 
Moreover, understanding  when extinction events can occur is relevant in the biological setting as such events may imply that the production of a certain protein has halted, or that a certain important reactant is eventually consumed completely. 
%On the other end, the existence of a stationary distribution with support on the whole state space implies survival of the biological system, which reaches a stationary regime.
However, we demonstrate through the analysis of a number of examples that a series of expected connections do not hold in general, and therefore discourage interested researchers from assuming them, or trying to prove them. In particular, we prove Conjecture 3.7 in \cite{AEJ:ACR} to be false. 
To better understand the work carried out here, we briefly describe some of the work in \cite{AEJ:ACR}.

In \cite{AEJ:ACR}, an interesting connection between extinction events and systems with \emph{absolute concentration robustness} (ACR)  is unveiled. ACR systems are deterministic reaction systems in which at least one chemical species has the same value at every positive equilibrium of the system. Such species are called absolute concentration robust (ACR) species. As an example, consider the network
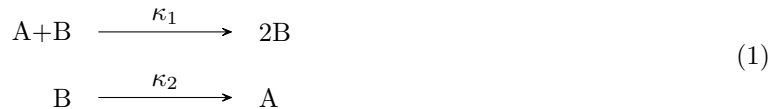
\begin{equation}\label{eq:toy_model}
 \begin{split}
 \schemestart
 A+B \arrow{->[$\kappa_1$]} 2B
 \arrow(@c1.south east--.north east){0}[-90,.25]
 B \arrow{->[$\kappa_2$]} A   
\schemestop
\end{split}
\end{equation}
Under the assumption of mass-action kinetics, the considered system is ACR: the species $A$ has the value $\kappa_2/\kappa_1$ in all the positive equlibria of the system. When modeled stochastically, the reaction $B\to A$ can take place until no molecules of $B$ are left. When this happens, no reaction can take place anymore as each of the reactions in \eqref{eq:toy_model} requires a molecule of $B$ as a reactant. We call this an \emph{extinction event} and note that it eventually occurs with a probability of one, regardless of the rate constants $\kappa_1, \kappa_2\in \RR_{>0}$.  This differing qualitative behavior between the two models (robustness for the ODE and eventual extinction for the stochastic) was studied in \cite{AEJ:ACR} and was proven to be general, in the following sense.  In \cite{SF:ACR}, Shinar and Feinberg provided sufficient necessary conditions  for a deterministic reaction system to be ACR.  However, in  \cite{AEJ:ACR}  it was shown that  the stochastic model will, with a probability of one, undergo an extinction event if the reaction network satisfies those same conditions and also has a positive conservation relation.  Moreover, the eventual extinction holds regardless of choice of rate constants or initial condition. % (as the total mass of the molecules present), then the corresponding stochastic model has an almost sure extinction event.

The assumptions of \cite{AEJ:ACR} (and by extension \cite{SF:ACR}) seem technical and do not unveil a clear reason for why the family of stochastic models considered have almost sure extinction events. Since under the same assumptions the corresponding deterministic system is ACR, it seems natural to conjecture that (i) absolute concentration robustness of the associated deterministic model, and (ii) the existence of a positive conservation relation, are sufficient to  imply almost sure extinction for the stochastic model. This is the content of Conjecture 3.7 in \cite{AEJ:ACR}. 

Belief in Conjecture 3.7 in \cite{AEJ:ACR} becomes even stronger when we realize that conservative ACR deterministic reaction systems often have boundary equilibria.   For example, in the model \eqref{eq:toy_model}, the deterministic model has boundary equilibria of the form $(a,0)$, which are attracting for initial conditions where the total mass is lower than $\kappa_2/\kappa_1$. In the stochastic model the species $B$ is eventually completely consumed, showing that the behaviors of the two models do have some connection.  Returning to the general setting, the existence of boundary equilibria occurs often for conservative ACR models since 
on certain invariant regions the total ``mass''  of the species (as determined by the positive conservation law) is  \emph{strictly less than}  the ACR value, implying there can be no positive equilibria in that invariant region.    Since the invariant regions are compact due to the existence of a mass conservation, the ODE solution is often attracted to the boundary. Intuitively, such attraction might indicate a propensity of the stochastic model to reach the boundary, and get absorbed. 
 
 However, in section \ref{sec:ACR}, we will show with non-trivial examples that the intuition the conjecture is based upon does not hold in general. In particular, we show that no assumption in the main result of \cite{AEJ:ACR} can be eliminated.  Moreover, we do so with bimolecular examples (which are the most commonly used examples in the biological setting).
%even if the absolute concentration robustness of the system is conserved, and even if the additional hypothesis of bimolecularity is introduced (which is very common in concrete biological models).
 In section \ref{sec:stable_boundary}, we will further explore the connection (or rather, the lack thereof) between extinction events of the stochastic model and equilibria of the associated deterministic model, and put to rest the common misconceptions that (i) a complete lack of positive equilibria in  a conservative  deterministic model implies the occurrence of an extinction event in the stochastic model, and (ii) that positive recurrence of all the states of the stochastic model implies the existence of a positive equilibrium of the associated deterministic model.

The outline of the remainder of the paper is as follows.  In section \ref{sec:notation}, we provide the necessary material on notation, and the formal introduction of the relevant mathematical models.  In section \ref{sec:correspondence}, we will state the classical result from \cite{kurtz:classical, kurtz:strong} precisely.  This results provides a connection between the behavior of the stochastic and deterministic models on compact time intervals.  We will then provide  examples which demonstrate a discrepancy in the long-term behavior of the models (in terms of explosions and positive recurrence of the stochastic model with respect to blow-ups and compact trajectories of the deterministic model).  In section \ref{sec:extinction}, we provide the necessary definitions relating to extinctions in the present context.  In particular, we point out that extinctions should refer to both species and reactions, as opposed to just the counts of species. Finally, in sections \ref{sec:ACR} and  \ref{sec:stable_boundary}  we provide our main results and analyze a number of examples as described in the previous paragraph.

\section{Necessary Background and Notation}
\label{sec:notation}
\subsection{Notation}
 We denote the non-negative integer numbers by $\ZZ_{\geq0}$. We also denote the non-negative and the positive real numbers by $\RR_{\geq0}$ and $\RR_{>0}$, respectively. Given a real number $a$, we denote by $|a|$ its absolute value. For any real vector $v\in\RR^d$, we denote its $i$th entry by $v_i$, and we use the notation
 $$\|v\|_1=\sum_{i=1}^d |v_i|.$$
 We further write $v>0$ and say that $v$ is positive if every entry of $v$ is positive. Moreover, given two real vectors $v,w\in\RR^d$, we write $v\geq w$ if $v_i\geq w_i$ for all $1\leq i\leq d$. %We also denote by
%  $$\mathbbm{1}_{v\geq w}=\begin{cases}
%                           1 & \text{if }v\geq w\\
%                           0 & \text{otherwise.}
%                          \end{cases}$$
 Finally, given a set $\mathcal{A}$, we denote its cardinality by $|\mathcal{A}|$.

\subsection{Basic definitions of reaction network theory}\label{sec:background}
%  We denote the natural numbers including 0 by $\ZZ_{\geq0}$, that is $\mathbb N = \{0,1,2,\dots\}$. For any real vector $v$, we denote its $i$th entry by $v_i$. We will write $v>0$ if every entry of $v$ is positive. We denote by $[v]$ the vector of the floor functions of the entries of $v$; that is, $[v]_i = \floor{v_i}$. For any real vector $\alpha$ of the same size as $v$, and for $N>0$, we denote by $N^{\alpha}v$ the  vector satisfying
%  \[
%  	\pr{N^{\alpha}v}_i=N^{\alpha_i}v_i.
%  \]
%  We will denote by $\|v\|$ the euclidean norm of the vector, by $\|v\|_1$ its \emph{$L^1$-norm} and by $\|v\|_\infty$ its \emph{$L^\infty$-norm}, that is
%  \[
%  	\|v\| = \sqrt{\sum_i v_i^2}, \qquad 
% 	 \|v\|_1=\sum_i|v_i|, \quad\text{and}\quad \|v\|_\infty=\max_i|v_i|.
% \]
%  For two vectors $v$ and $w$ of the same dimension, we write $v<w$, $v\leq w$, $v>w$ or $v\geq w$ if the inequality holds component-wise. Furthermore, for any set $A$ we will indicate by $|A|$ its cardinality and by $\mathbbm{1}_A$ its indicator function.
%  Finally, for $a,b \in \mathbb{R}$, we denote $a\wedge b = \min\{a,b\}$ and $a\vee b = \max\{a,b\}$.
 
%  Here we give some basic definitions from chemical reaction network theory, see for example \cite{erdi:mathematical_models,feinberg:lectures} for a more detailed introduction.
 
 A \emph{reaction network} is a triple $\G=(\Sp,\C,\R)$. $\Sp$ is a finite non-empty set of symbols, referred to as \emph{species}, and $\C$ is a finite non-empty set of linear combinations of species with non-negative integer coefficients, referred to as \emph{complexes}. After ordering the set of species, the $i$th species of the set can be identified with the vector $e_i\in\RR^{|\Sp|}$, whose $i$th entry is 1 and whose other entries are zero. It follows that any complex $y\in\C$ can be identified with a vector in $\ZZ^{|\Sp|}_{\geq0}$, which is the corresponding linear combination of the vectors $e_i$. Finally, $\R$ is a non-empty subset of $\C\times\C$, whose elements are called \emph{reactions}, such that for any $y\in\C$, $(y,y)\notin\R$. Following the common notation, we will denote any element $(y,y')\in\R$ by $y\to y'$.
%  in which case we then call $y_r$ the \emph{source complex} and $y_r'$ the \emph{product complex} of that reaction. It is possible that a complex $y\in\C$ is the source (product) complex of different reactions, and that it is both the source complex of one reaction and the product complex of another reaction. 
 We require that every species in $\Sp$ appears in at least one complex, and that every complex in $\C$ appears as an element in at least one reaction. Note that under this condition, a reaction network is uniquely determined by the set of reactions $\R$.
 
 In this paper, species will always be letters and will be alphabetically ordered. %To increase readability, to refer to the entry of a vector in $\RR^{|\Sp|}$ relative to the $i$th species (say the second species, denoted by $B$) we will use as index both $i$ and the species itself (we will therefore  use both the notation $y_i$ and $y_B$), depending on what is more convenient in any particular case. 
 
 A directed graph can be associated in a very natural way to a reaction network by considering the set of complexes as nodes and the set of reactions ad directed edges. Such a graph is called a \emph{reaction graph}. Usually, a reaction network is presented by means of its reaction graph, which defines it uniquely. By using the reaction graph, we can further define the \emph{terminal} complexes as those nodes that are contained in some strongly connected component of the graph, that is those nodes $y$ such that for any directed path from $y$ to another complex $y'$ there exists a directed path from $y'$ to $y$. We say that a complex is \emph{non-terminal} if it is not terminal.
 
 We define the \emph{stoichiometric subspace} of a reaction network as
 $$S=\spann_{\RR}\{y'-y\,:\,y\to y'\in\R\}.$$
 We also let $\ell$ denote the number of connected components of the reaction graph (or \emph{linkage classes} of the network) and define the \emph{deficiency} of the network as
 $$\delta=|\C|-\ell-\dim S.$$
 The geometric interpretation of the deficiency is not clear from the above definition, see \cite{Gun2003} for more detailes on this object.
 
 We say that a vector $v\in\ZZ^{|\Sp|}$ is a \emph{conservation law} if it is orthogonal to the stoichiometric subspace $S$. We say that a network is \emph{conservative} if there is a  positive coservation law: this means that it is possible to assign a mass to the molecules of each chemical species such that the total mass (i.e.\ the sum of the masses of all the molecules present) is conserved by the occurrence of every reaction.
 
 We say that a reaction network is \emph{bimolecular} if $\max_{y\in\C}\|y\|_1\leq 2$. Many biological models fall into this category, since it is often the case that at most two molecules react at a time. 
 
Finally, we associate with each reaction $y\to y'\in\R$ a positive real number $\kappa_{y\to y'}$, called a \emph{rate constant}. A reaction network with a choice of rate constants is called a \emph{mass-action system}, and a stochastic or a deterministic dynamics can be associated with it, as described later. A mass-action system is usually presented by means of the reaction graph, where the reactions have been labelled with the corresponding rate constant. An example of this can be found in \eqref{eq:toy_model}.
%More general kinetics than mass-action can be considered (for example, Michaelis-Menten kinetics, Hill kinetics or power law kinetics), but in the present paper we are only concerned about mass-action kinetics, which corresponds to the hypothesis that all molecules are well-mixed in the environment.

% $$\lambda_{y\to y'}:\RR^{|\Sp|}_{\geq0}\to\RR_{\geq0}.$$ 
% The set consisting of  these function $\K=\{\lambda_{y\to y'}\}_{y\to y'\in\R}$ is referred to as \emph{kinetics}, and the functions $\lambda_{y\to y'}$ are called \emph{rate functions}, or \emph{intensity functions}, or \emph{propensity functions}. The pair $\Sy=(\G,\K)$ is a \emph{reaction system}, which can be stochastically or deterministically modeled, as explained below.
 
 \subsubsection{Stochastic Model}
 In a \emph{stochastic mass action system}, the evolution in time of the copy-numbers of molecules of the different chemical species is considered. Specifically, the  copy-numbers of the molecules of different chemical species at time $t\geq0$ form a vector $X(t)\in\ZZ_{\geq0}^{|\Sp|}$. The process $X$ is assumed to be a continuous time Markov chain. Specifically, for any two states $x, x'\in\ZZ_{\geq0}^{|\Sp|}$ the transition rate from $x$ to $x'$ is given by
 $$q(x,x')=\sum_{\substack{y\to y'\in\R\\y'-y=x'-x}}\lambda_{y\to y'}(x),$$
 where
 $$\lambda_{y\to y'}(x)=\kappa_{y\to y'}\prod_{i=1}^{|\Sp|}\prod_{j=0}^{y_i-1}(x_i-j)\quad\text{for }x\in\ZZ_{\geq0}^{|\Sp|}$$
 is the \emph{stochastic mass action rate function} of $y\to y'$. Note that $\lambda_{y\to y'}(x)>0$ if and ony if $x\geq y$, which prevent the entries of the process $X$ from becoming negative. Also, note that the process is confined within an \emph{stoichiometric compatibility class}, that is for every $t\geq0$
 $$X(t)\in \{X(0)+v\,:\, v\in S\}\cap \ZZ_{\geq0}^{|\Sp|}.$$
 
 It is worth noting that the stochastic mass-action kinetics follows from the assumption that the molecules of the different species are well-mixed, so the propensity of each reaction to take place is proportional to the number of possible sets of molecules that can give rise to an occurrence of the reaction. Other kinetics may arise in different scenarios, but in the present paper we are only concerned with mass action systems.
 
 \subsubsection{Deterministic Model}
 In a  \emph{deterministic mass action system} the evolution in time of the concentrations of the different chemical species is modeled. We consider the concentrations of the different chemical species at time $t\geq0$ as a vector $z(t)\in\RR_{\geq0}^{|\Sp|}$. It is then assumed that the function $z$ is a solution to the Ordinary Differential Equation (ODE)
  \begin{equation}
  \label{eq:dma1}
  	\frac{d}{dt}z(t) = g(z(t)),
  \end{equation}
 where
 $$g(x)=\sum_{y\to y'\in\R}(y'-y)\kappa_{y\to y'}\prod_{i=1}^{|\Sp|}x_i^{y_i}\quad\text{for }x\in\RR_{\geq0}^{|\Sp|}$$
 is the \emph{deterministic mass action species formation rate}. As in the stochastic case, the solution $z$ is confined within a stoichiometric compatibility class, meaning that for any $t\geq0$
 $$z(t)\in \{z(0)+v\,:\, v\in S\}\cap \RR_{\geq0}^{|\Sp|}.$$
 Finally, as for stochastic models, the choice of mass action kinetics corresponds to the assumption that the molecules of the different species are well-mixed. Other kinetics (such as Michaelis-Menten kinetics, Hill kinetics, power law kinetics) are considered in the literature, but are not dealt with in this paper.

 \section{Correspondences between the two modeling regimes and known discrepancies}
 \label{sec:correspondence}
 
 The aim of this section is to describe why it was believed that certain properties of the deterministic mass action system would imply the occurrence of extinction events for the stochastically modeled mass action system. Here, we briefly describe or give reference to some known connections between the two modeling regimes, and provide some warnings in the form of examples of discrepancies between the two models.
 
 \subsection{Connections}
 
 The first connection found between stochastic and deterministic models dates back to \cite{kurtz:classical}. The situation considered is the following: a reaction network $\G$ is given. It is assumed that the volume $V$ of a container where the chemical transformations occur is increased. The propensity of a reaction to occur changes with the volume, depending on how many molecules are needed (i.e.\ need to collide) for the reaction to take place. Specifically, the rate constants scale with the volume as
 $$\kappa^V_{y\to y'}=V^{1-\|y\|_1}\kappa_{y\to y'},$$
 for some fixed positive constants $\kappa_{y\to y'}$. A family of continuous-time Markov chain $\{X^V\}_{V}$ is then defined, with $X^V$ being the process associated with the stochastic mass action system with rate constants $\kappa^V_{y\to y'}$. Then, the following holds \cite{kurtz:classical,kurtz:strong}.
 \begin{theorem}\label{thm:classical_scaling}
  Assume that for a fixed positive state $z_0\in\RR_{>0}^{|\Sp|}$ and for all $\varepsilon>0$ we have
  $$\lim_{V\to\infty}P\Big(\Big|V^{-1}X^V(0)-z_0\Big|>\varepsilon\Big)=0.$$
  Moreover, assume that the solution $z$ of the ODE \eqref{eq:dma1} with $z(0)=z_0$ is unique and is defined up to a finite fixed time $T>0$. Then, for any $\varepsilon>0$
  $$\lim_{V\to\infty}P\Big(\sup_{t\in[0,T]}\Big|V^{-1}X^V(t)-z(t)\Big|>\varepsilon\Big)=0.$$
 \end{theorem}
 Roughly speaking, the theorem states the rescaled stochastic processes $X^V$ converge path-wise to the ODE solution of the deterministic mass action system, over compact intervals of time. The theorem also holds for more general kinetics than mass action kinetics, as long as the rate functions $\lambda_{y\to y'}$ are locally Lipschitz \cite{kurtz:strong}. Theorem \ref{thm:classical_scaling} has been extended for the multiscale setting \cite{KK:multiscale}. 
 
 It is interesting to note that  model reduction techniques over compact intervals of time also work for both the deterministic and stochastic models in the exact same manner. In particular, we refer to the assumptions under which intermediate species can be eliminated from a multiscale model, and to the description of the resulting simplified model \cite{CW:intermediate_stoc, CW:intermediate_det}.
 
 \subsection{Known discrepancies}
 
 Here we cite some known examples showing that compactness of the trajectories of a deterministic mass action system do not imply positive recurrence or even regularity (that is, lack of explosions) for the associated stochastic model. Moreover, we demonstrate that a blowup of the deterministic mass-action system does not in general imply explosions for the associated stochastic mass action system.
 
In \cite{ACKN:endotactic} it is shown that the mass action system
\begin{equation}\label{ex:endotactic_trans}
 \begin{split}
\schemestart
 0\arrow{->[$\kappa_1$]}2A+B\arrow{->[$\kappa_2$]}4A+4B\arrow{->[$\kappa_3$]}A
 \schemestop  
 \end{split}
\end{equation}
is transient (i.e.\ all states are transient) when stochastically modeled, for all choices of rate constants. For the deterministic mass action system, however, for any choice of rate constants there exists a compact set $K$ such that for all initial conditions $z(0)$, a $t^*>0$ exists with $z(t)\in K$ for all $t>t^*$ (the system is said to be \emph{permanent}, a property that in this case follows from the network being \emph{strongly endotactic} \cite{GMS:geometric}). Moreover, in \cite{ACKN:endotactic} it is also shown that for any choice of rate constants the mass action system
\begin{equation}\label{ex:endotactic_expl}
\begin{split}
 \schemestart
 0\arrow{->[$\kappa_1$]}2A\arrow{->[$\kappa_2$]}4A+B\arrow{->[$\kappa_2$]}6A+4B\arrow{->[$\kappa_3$]}3A
 \schemestop 
\end{split}
\end{equation}
is explosive (in the sense of \cite{norris:markov}) when stochastically modeled, while the deterministic mass action system is permanent as for \eqref{ex:endotactic_trans} (which again follows because the network is strongly endotactic).

Since Theorem~\ref{thm:classical_scaling} holds, we expect the time of drifting towards infinty of the processes $X^V$ associated with \eqref{ex:endotactic_trans} to increase with $V$. Similarly, the time until explosion of the processes  $X^V$ associated with \eqref{ex:endotactic_expl} necessarily  tends to infinity, as $V\to\infty$.
 
 We have shown examples of mass action systems that are somehow well behaved if deterministically modeled, and transient or explosive if stochastically modeled. For completeness, we also present here  an example of a mass action system that is positive recurrent (i.e.\ all states are positive recurrent) if stochastically modeled, while the associated determinisitc ODE solution has blow-ups for any positive initial condition. The system is discussed in \cite{ACKK:explosion} and is the following.
 
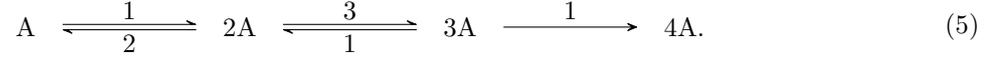
\begin{equation}\label{ex:blow_pos_rec}
  \begin{split}
   \schemestart
   A\arrow{<=>[1][2]}2A\arrow{<=>[3][1]}3A\arrow{->[1]}4A.
   \schemestop
  \end{split}
 \end{equation}
It is worth citing here a very similar example, also discussed in \cite{ACKK:explosion}, where the behaviour of the two modeling regimes is similar. Consider the mass action system
 
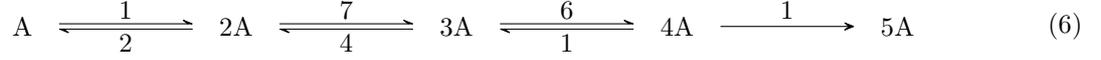
\begin{equation}\label{ex:no_blow_expl}
  \begin{split}
   \schemestart
   A\arrow{<=>[1][2]}2A\arrow{<=>[7][4]}3A\arrow{<=>[6][1]}4A\arrow{->[1]}5A
   \schemestop
  \end{split}
 \end{equation}
It is shown in \cite{ACKK:explosion} that the corresponding stochastic mass action system is explosive for any positive initial condition $X(0)$, and the associated deterministic mass action system has a blow up for any positive initial condition $z(0)$.

Other examples of coincidence between the long-term behaviour of the stochastic and deterministic models of mass action systems are given by the family of \emph{complex balanced systems}, and are studied in \cite{ACK:poisson, CW:poisson, CJ:graphical}. The existence of such connections and Theorem~\ref{thm:classical_scaling} contributed to the formulation of the conjecture that ``mass action systems with ACR species when deterministically modeled undergo an extinction event when stochastically modeled'', for the reasons concerning boundary equilibria of the deterministic model already discussed in the Introduction. However, we have shown in this section that the long-term dynamics of stochastically and deterministically modeled mass action systems can differ greatly.

 \section{Extinction} 
 \label{sec:extinction}
In this section, we will formally describe  what is meant by the term ``extinction'' in the present context.  We begin with the following standard definitions.
 \begin{definition}
  Consider a stochastic mass action system. We say that
  \begin{itemize}
   \item a state $x'$ is \emph{reachable} from $x$ if for some $t>0$
   $$P(X(t)=x'\, |\, X(0)=x)>0;$$
   \item a set $\Gamma\subseteq \ZZ_{\geq0}^{|\Sp|}$ is \emph{reachable} from $x$ if for some $t>0$
   $$P(X(t)\in\Gamma \, |\, X(0)=x)>0;$$
   \item a set $\Gamma\subseteq \ZZ_{\geq0}^{|\Sp|}$ is \emph{closed} if for all $t>0$
   $$P(X(t)\in\Gamma \, |\, X(0)\in\Gamma)=1;$$
   \item a reaction $y\to y'\in\R$ is \emph{active} at a state $x$ if $\lambda_{y\to y'}(x)>0$.
   \item a set $\Gamma\subseteq \ZZ_{\geq0}^{|\Sp|}$ is an \emph{extinction set} for the reaction $y\to y'\in\R$ if $\Gamma$ is closed and $y\to y'$ is not active at any state of $\Gamma$.
%    \item a reaction $y\to y'\in\R$ is \emph{exhausted} if there exists a closed set $\Gamma$ such that $\lambda_{y\to y'}(x')=0$ for all $x'\in\Gamma$, and $X(t)\in\Gamma$ for some $t\geq0$.
   \end{itemize}
 \end{definition}
 \begin{definition}\label{def:extinction}
  Consider a stochastic mass action system. We say that the process $X$ undergoes an \emph{extinction event} at time $t^*>0$ if there is a reaction $y\to y'\in\R$ and an extinction set $\Gamma$ for $y\to y'$ such that $X(t^*)\in\Gamma$ and $X(t^*-) \notin \Gamma$.
 \end{definition}
 The meaning of the above definition is the following: at a certain time $t^*$ the copy-number of some chemical species (or a set of chemical species) gets so low that a certain reaction can not occur anymore, and the loss is irreversible.

 In connection with Definition~\ref{def:extinction}, we give some useful results.
 \begin{proposition}\label{prop:consecutio}
  Consider a stochastic mass action system, and two reactions $y\to y', \tilde y\to \tilde y'\in\R$, with $y'\geq \tilde y$ . Assume that no extinction set for $y\to y'$ is reachable from the state $x$. Then, no extinction set for $\tilde y\to \tilde y'$ is reachable from $x$.
 \end{proposition}

 \begin{proof}
  Since no extinction set for $y\to y'$ is reachable from $x$,  any closed set $\Gamma$ that is reachable from $x$ contains a state $x'$ such that $\lambda_{y\to y'}(x')>0$, or equivalently $x'\geq y$. Hence, the state $x'+y'-y$ is reachable from $x'$, and since $\Gamma$ is closed it follows that $x'+y'-y\in\Gamma$. Moreover, since $x'\geq y$ and  $y'\geq \tilde y$ we have $x'+y'-y\geq y' \geq \tilde y$, which implies that $\lambda_{\tilde y\to \tilde y'}(x'+y'-y)>0$ and that $\Gamma$ is not an extinction set for $\tilde y\to\tilde y'$. This concludes the proof.
 \end{proof}

 \begin{proposition}\label{prop:dominatio}
  Consider a stochastic mass action system, and two reactions $y\to y', \tilde y\to \tilde y'\in\R$, with $y\geq \tilde y$. Assume that no extinction set for $y\to y'$ is reachable from the state $x$. Then, no extinction set for $\tilde y\to \tilde y'$ is reachable from $x$.
 \end{proposition}

\begin{proof}
 As in the proof of Proposition~\ref{prop:consecutio}, since no extinction set for $y\to y'$ is reachable from $x$,  any closed set $\Gamma$ that is reachable from $x$ contains a state $x'$ with $\lambda_{y\to y'}(x')>0$, which is equivalent to $x'\geq y$. Since $ y\geq \tilde y$, we have $x'\geq \tilde{y}$, which is equivalent to $\lambda_{\tilde y\to \tilde y'}(x')>0$.  This concludes the proof.
\end{proof}

\section{Extinction and ACR systems}\label{sec:ACR}

We give here the formal definition of system that has Absolute Concentration Robustness (ACR) . The definition is purely in terms of deterministic systems, and the precise connection with extinction events will be described later.

\begin{definition}\label{def:ACR}
 Consider a deterministic mass action system. The system is said to be \emph{Absolute Concentration Robust} (ACR) if there exists an index $1\leq i\leq |\Sp|$ and a real number $u\in\RR_{>0}$ such that all $c>0$ with $g(c)=0$ satisfy $c_i=u$. In this case, the $i$th species is called an \emph{ACR species} with \emph{ACR value} $u$.
\end{definition}

Note that, by definition, all deterministic mass action systems with no positive equilibria, or with only one positive equilibrium are ACR. Specifically, in these cases all species are ACR. However, such degenerate cases elude the sense of the definition of ACR models, which captures an important biological property: whenever the system is at a positive equilibrium (and there could be many positive equilibria), some special chemical species are always expressed at the same level, and are therefore ``robust'' to environmental changes.

Structural sufficient conditions for a model to be ACR can be found in the following result, due to Feinberg and Shinar \cite{SF:ACR}.
\begin{theorem}\label{thm:ACR_det}
 Consider a deterministic mass action system and assume
 \begin{enumerate}
  \item there exists at least one $c>0$ with $g(c)=0$;
  \item the reaction network has deficiency 1;
 \item there are two non-terminal complexes $y\neq y'$ such that only the $i$th entry of $y'-y$ is different from 0.
 \end{enumerate}
 Then, the $i$th species is ACR.
\end{theorem}
Note that the assumption on the existence of a positive equilibrium is not needed for the theorem to hold, since if there were none then all the species would  automatically be ACR by Definition~\ref{def:ACR}. However, the assumption was included because this degenerate case was not considered in \cite{SF:ACR}.

The connection with extinction events is given by the following result, due to Anderson, Enciso, and Johnston  \cite{AEJ:ACR}.
\begin{theorem}\label{thm:ACR_stoc}
 Consider a stochastic mass action system and assume
 \begin{enumerate}
   \item there exists at least one $c>0$ with $g(c)=0$;
  \item\label{part:def} the reaction network has deficiency 1;
  \item\label{part:non_term} there are two non-terminal complexes $y\neq y'$ such that only the $i$th entry of $y'-y$ is different from 0;
  \item\label{part:cons} the reaction network is conservative.
 \end{enumerate}
 Then, an extinction event occurs almost surely. In particular, with a probability of one the process $X$ enters a closed set $\Gamma\subset \ZZ^d_{\geq0}$ such that $\lambda_{y\to y'}(x)=0$ for all $x\in\Gamma$ and for all $y\to y'\in\R$ with $y$ being a non-terminal complex.
\end{theorem}

As an example of how to apply Theorems~\ref{thm:ACR_det} and \ref{thm:ACR_stoc}, consider the mass-action system \eqref{eq:toy_model}. The deficiency is $\delta=4-2-1=1$, the two non-terminal complexes $A+B$ and $B$ are such that their difference is $A$, and there is at least one $c>0$ with $g(c)=0$. Hence, by Theorem~\ref{thm:ACR_det}, when deterministically modeled, the mass action \eqref{eq:toy_model} is ACR, and in particular the species $A$ appears with the same value at every positive equilibrium of the system. In fact, it can be easily checked that all positive equilibria are of the form
$\left(\frac{\kappa_2}{\kappa_1}, \beta\right)$
for some $\beta>0$, where the species are ordered as $(A,B)$. Moreover, the reaction network of \eqref{eq:toy_model} is conservative, since $(1,1)$ is a positive conservation law. Hence, by Theorem~\ref{thm:ACR_stoc}, the stochastically modeled mass action system will, with a probability of one, undergo  an extinction event.  In particular, both reactions $B\to A$ and $A+B\to 2B$ cannot occur after the extinction. Since the total mass is conserved, this can only mean that the molecules of $B$ are eventually completely consumed.

Note that under the assumptions of Theorem~\ref{thm:ACR_stoc}, the deterministically modeled mass action system is ACR. Since the assumption of Theorem~\ref{thm:ACR_stoc} are quite technical in nature, as discussed in the Introduction it was thought for a long time that the real reason leading to the extinction event in the stochastic mass action system was the associated deterministic mass action system being ACR. This seemed plausible since ACR systems often exhibit attracting boundary equilibria. As an example, if the total mass of the deterministic mass action system \eqref{eq:toy_model} is lower then the ACR value of $A$, that is if $\|z(0)\|<\kappa_2/\kappa_1$, then the ODE solution is confined within a compact stoichiometric compatibility class with no positive equilibria, and is eventually attracted by the boundary equilibrium $(z(0),0)$. This resembles what happens with the stochastically modeled system.

In this section we show that such intuition is not correct.  We do so by proving that if you remove any of the technical assumptions \eqref{part:def}, \eqref{part:non_term}, and \eqref{part:cons} from Theorem~\ref{thm:ACR_stoc}, while maintaining the absolute concentration robustness of the associated deterministic mass action system, then the result no longer holds. Moreover, the previous sentence holds even if we add the additional constraint that the reaction network  be bimolecular.

\subsection*{Example 1: assumption~\eqref{part:cons} cannot be removed}

In \cite{AEJ:ACR}, the authors realized that assumption~\eqref{part:cons} could not be removed from the statement of Theorem~\ref{thm:ACR_stoc}, and proved that with the following example. Consider the reaction network
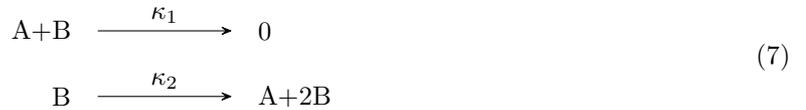
\begin{equation}\label{eq:example0}
 \begin{split}
 \schemestart
 A+B \arrow{->[$\kappa_1$]} 0
 \arrow(@c1.south east--.north east){0}[-90,.25]
 B \arrow{->[$\kappa_2$]} A+2B   
\schemestop
\end{split}
\end{equation}
The mass action system satisfies the assumptions of Theorem~\ref{thm:ACR_det} and the species $A$ is ACR, with ACR equilibrium $\kappa_2/\kappa_1$. The only assumption of Theorem~\ref{thm:ACR_stoc} that is not met is \eqref{part:cons}, as the network is not conservative. Indeed, the stoichiometric subspace is 
$$S=\left\{\begin{pmatrix}s\\s\end{pmatrix}\,:\,s\in\RR\right\}$$
and no positive vector is contained in
$$S^\perp=\left\{\begin{pmatrix}s\\-s\end{pmatrix}\,:\,s\in\RR\right\}.$$
In \cite{AEJ:ACR}, it is shown that the stochastically modeled mass action system does not undergo any extinction event if $X_1(0)<X_2(0)$ (remember we assume the species to be alphabetically ordered). In this case, the stationary distribution is computed by using birth and death process techniques, and every state is shown to be positive recurrent.

We can modify \eqref{eq:example0} to obtain the bimolecular mass action system
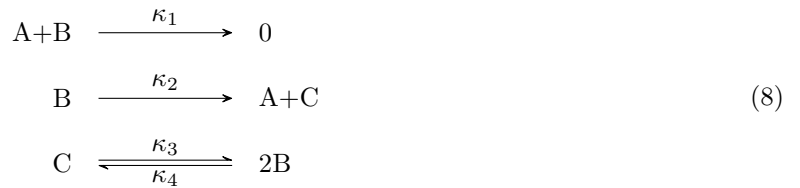
\begin{equation}\label{eq:example0_bi}
 \begin{split}
 \schemestart
 A+B \arrow{->[$\kappa_1$]} 0
 \arrow(@c1.south east--.north east){0}[-90,.25]
 B \arrow{->[$\kappa_2$]} A+C   
  \arrow(@c3.south east--.north east){0}[-90,.25]
 C \arrow{<=>[$\kappa_3$][$\kappa_4$]} 2B   
\schemestop
\end{split}
\end{equation}
The mass action system \eqref{eq:example0_bi} still satisfies the assumptions of Theorem~\ref{thm:ACR_det}, and it can therefore be concluded that the species $A$ is ACR.  In fact, the ACR value for $A$ is still $\kappa_2/\kappa_1$.  Again, the only assumption of Theorem~\ref{thm:ACR_stoc} that is not satisfied is \eqref{part:cons}, since it can be checked that the orthogonal to the stoichiometric subspace is 
$$S^\perp=\left\{\begin{pmatrix}s\\-s\\-2s\end{pmatrix}\,:\,s\in\RR\right\},$$
which does not contain any positive vector. Define the conserved quantity $m=X_1(0)-X_2(0)-2X_3(0)$ and assume that $m<0$. We will show that no extinction set for any reaction exists, which in turn proves that no extinction event can occur.

Note that since $X$ is confined within a stoichiometric compatibility class, for any $t\geq0$
$$X_1(t)-X_2(t)-2X_3(t)=m<0.$$
 
 Assume that from $X(0)$ a state $x'$ can be reached  with $\lambda_{A+B\to 0}(x')=0$. Then, either 
 \begin{itemize}
 \item[(i)] $x'_2=0$, or
 \item[(ii)]  $x'_2>0$ and $x'_1=0$.
 \end{itemize}
  Suppose we are in case (i).  Then, since
 $$m=x'_1-2x'_3<0,$$
 we must have $x'_3>0$. Hence the reaction $C\to2B$ can take place, in which case at least one molecule of $B$ is present.  At this point, either $A+ B \to 0$ is active, or we are in case (ii).
Hence, assume that (ii) holds.  Then $B\to A+C$ can occur, which can then be followed by $C \to 2B$.  After these two reactions take place, both a molecule of $A$ and a molecule of $B$ are necessarily present, so a state is reached where the reaction $A+B\to 0$ is active.  
 
 Combining all of the above, we have proven that  no extinction set for $A+B\to 0$ is reachable from $X(0)$.  By applying Proposition~\ref{prop:dominatio}, it follows that the same holds for the reaction $B\to A+C$. Then, by Proposition~\ref{prop:consecutio} we have that no extinction set for $C\to 2B$ is reachable from $X(0)$, and finally by applying Proposition~\ref{prop:consecutio} we conclude that the same holds for $2B\to C$. In conclusion, no extinction event can occur with the chosen initial conditions.

\subsection*{Example 2: assumption~\eqref{part:def} cannot be removed}

Consider the bimolecular mass action system
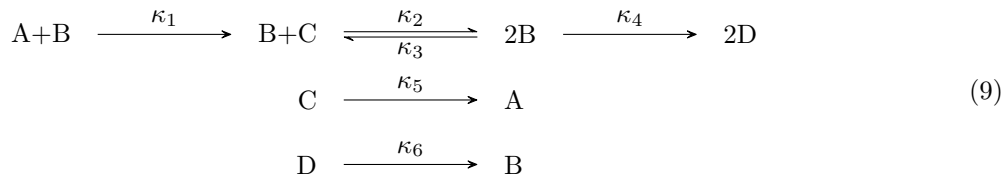
\begin{equation}\label{ex:no_def_1}
 \begin{split}
  \schemestart
   A+B\arrow{->[$\kappa_1$]}B+C\arrow{<=>[$\kappa_2$][$\kappa_3$]}2B\arrow{->[$\kappa_4$]}2D
   \arrow(@c2.south east--.north east){0}[-90,.25]
   C\arrow{->[$\kappa_5$]}A
   \arrow(@c5.south east--.north east){0}[-90,.25]
   D\arrow{->[$\kappa_6$]}B
  \schemestop
  \end{split}
\end{equation}
The following holds.
\begin{description}
 \item[The species $A$ is ACR. Moreover, there exists at least one $c>0$ with $g(c)=0$.] Indeed, it can be checked that
 $$g(c)=\begin{pmatrix}
                 -\kappa_1c_1c_2+\kappa_5c_3\\
                 \kappa_2c_2c_3-\kappa_3c_2^2-2\kappa_4c_2^2+\kappa_6c_4\\
                 \kappa_1c_1c_2-\kappa_2c_2c_3+\kappa_3c_2^2-\kappa_5c_3\\
                 2\kappa_4c_2^2-\kappa_6c_4
                \end{pmatrix}$$
 is zero if and only if $c_2=c_3=c_4=0$ or
 $$c=\left(
      \frac{\kappa_3\kappa_5}{\kappa_1\kappa_2},
      s,
      \frac{\kappa_3}{\kappa_2}s,
      \frac{2\kappa_4}{\kappa_6}s^2
     \right)\quad\text{for some }s\in\RR_{>0}.
$$
%  \item[The boundary equilibria are attracting.] Indeed, if $x\in\RR^4$ is such that $c_2, c_3, c_4<\varepsilon$ for $\varepsilon>0$ small enough, then
 \item[The reaction network has deficiency 2.] This can be easily checked, since
 $$S=\spann_\RR\left\{
                    \begin{pmatrix}
                      -1\\ 0\\ 1\\ 0
                     \end{pmatrix},
                     \begin{pmatrix}
                      0\\ 1\\ -1\\ 0
                     \end{pmatrix},
                     \begin{pmatrix}
                      0\\ -1\\ 0\\ 1
                     \end{pmatrix}
\right\}$$
and
 $$\delta=|\C|-\ell-\dim S=8-3-3=2.$$
 \item[There are two non-terminal complexes $y\neq y'$ such that only the second entry of $y'-y$ is different from 0.] Indeed, the two complexes $B+C$ and $C$ are non-terminal and their difference is $B$. It is interesting to note that in this example the species $B$ is not the ACR species, so the conclusions of Theorem~\ref{thm:ACR_det} do not hold.
 \item[The reaction network is conservative.] Indeed, the vector $(1,1,1,1)$ is a positive conservation law.
%  \item[The reaction network is bimolecular.] It can be easily checked.
\end{description}
Hence, all assumptions of Theorem~\ref{thm:ACR_stoc} are fulfilled except for \eqref{part:def}, and the deterministically modeled mass action system is ACR. We will show that no extinction event can take place for the stochastic mass action system, if we choose an initial condition $X(0)$ with the minimal requirements that the conserved mass $m=X_1(0)+X_2(0)+X_3(0)+X_4(0)\geq 2$ and that $X_2(0)+X_4(0)\geq 1$. Specifically, we will show that there is no extinction set for any $y\to y'\in\R$ which is reachable from $X(0)$ . 

Assume that a state $x'$ with $\lambda_{A+B\to B+C}(x')=0$ is reachable from $X(0)$. Then, there are two cases:
\begin{itemize}
 \item We have $x'_2=0$. Note that the only reaction reducing the total number of $B$ and $D$ molecules is $2B\to B+C$, which decreases the total number by 1 but can only take place if at least 2 molecules of $B$ are present. Hence, at least one molecule of $B$ or $D$ is always present (because $X_2(0)+X_4(0)\geq 1$), which implies that $x'_2+x'_4\geq 1$. More specifically, from $x'_2=0$ it follows that $x'_4\geq1$, implying that the reaction $D\to B$ can take place. Hence, a state with at least one molecule of $B$ can be reached. Either $A+B\to B+C$ is active at this state, or the following case can be considered.
 \item We have $x'_2>0$ and $x'_1=0$.  We consider three further subcases:
 \begin{itemize}
 \item If $x'_3>0$, then a molecule of $A$ can be created by the occurrence of $C\to A$, which does not modify the number of molecules of $B$. Hence, a state can be reached where $A+B\to B+C$ is active.
 \item  If $x'_3=0$ and $x'_2\geq 2$, then a molecule of $B$ can be transformed into a molecule of $C$ through $2B\to B+C$, which only consumes one molecule of $B$. This subcase is therefore reduced to the previous one.
 \item If $x'_3=0$ and $x'_2=1$, then
 $$2\leq m=x'_1+x'_2+x'_3+x'_4=0+1+0+x'_4,$$
 which implies that $x'_4\geq 1$ and an additional molecule of $B$ can be created by the occurrence of $D\to B$. This subcase is therefore reduced to the previous one.
 \end{itemize}
\end{itemize}
It follows from the above analysis that no extinction set for $A+B\to B+C$ is reachable from $X(0)$. By consecutive applycations of Proposition~\ref{prop:consecutio}, it follows that the same occurs for all the other reactions as well.
% 
% \begin{itemize}
%  \item Assume that $\lambda_{A+B\to B+C}(X(t))=0$.  In the first case, following the argument detailed above, a molecule of $B$ can be  created without modifying the number of molecules of $A$, so we can just consider the second case without loss of generality.
%  
% 
%  
%  In any case, it is always possible to reach a state where the reaction $A+B\to B+C$ can take place.
%  \item Assume that $\lambda_{A+B\to B+C}(X(t))=0$. Then, either $X_2(t)=0$, or $X_2(t)>0$ and $X_1(t)=0$. In the first case, following the argument detailed above, a molecule of $B$ can be  created without modifying the number of molecules of $A$, so we can just consider the second case
%  
% \end{itemize}
% 
% There is no extinction event because two molecules of $B$ (or of $D$, which recreates $B$) are always present in the system.

\begin{remark}
 As noted above, in this case the two non-terminal complexes $B+C$ and $C$ differ in the second entry, but the second species (that is $B$) is not ACR. This is not what prevents the occurrence of an extinction event. A similar analysis as before can be conducted on the following bimolecular mass action system:
 
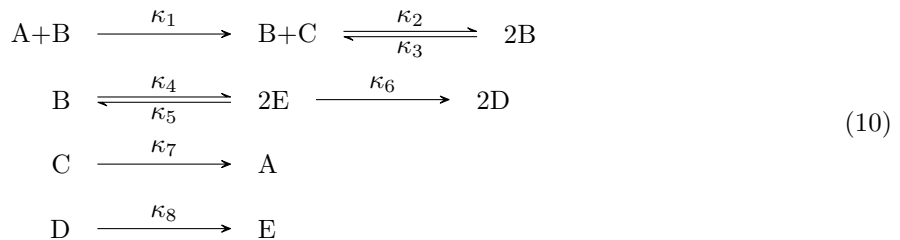
\begin{equation}\label{ex:no_def_1_remark}
 \begin{split}
  \schemestart
   A+B\arrow{->[$\kappa_1$]}B+C\arrow{<=>[$\kappa_2$][$\kappa_3$]}2B
   \arrow(@c1.south east--.north east){0}[-90,.25]
   B\arrow{<=>[$\kappa_4$][$\kappa_5$]}2E\arrow{->[$\kappa_6$]}2D
   \arrow(@c4.south east--.north east){0}[-90,.25]
   C\arrow{->[$\kappa_7$]}A
   \arrow(@c7.south east--.north east){0}[-90,.25]
   D\arrow{->[$\kappa_8$]}E
  \schemestop
  \end{split}
\end{equation}
For completeness, the analysis of \eqref{ex:no_def_1_remark} is carefully carried out in the Appendix. There, it is shown that the species $A$ is the only ACR species, with ACR value $(\kappa_3\kappa_7)/(\kappa_1\kappa_2)$. Moreover, the only two non-terminal complexes differing in one entry are $A+B$ and $B$, and they differ in the first entry. It is further shown that \eqref{ex:no_def_1_remark} satisfies all the assumptions of Theorem~\ref{thm:ACR_stoc}, except for \eqref{part:def}, and that no extinction event can occur provided that
\begin{align*}
 2X_1(0)+2X_2(0)+2X_3(0)+X_4(0)+X_5(0)&\geq4,\\
 2X_2(0)+X_4(0)+X_5(0)&\geq2.\\
\end{align*}
\end{remark}

\subsection*{Example 3: assumption \eqref{part:non_term} cannot be removed}

Consider the following bimolecular mass action system.
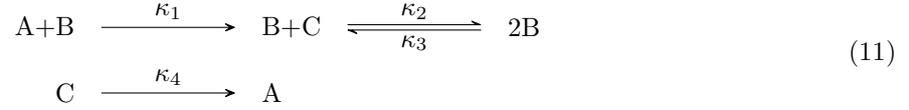
\begin{equation}\label{ex:non_non_terminal}
  \begin{split}
  \schemestart
   A+B \arrow{->[$\kappa_1$]} B+C \arrow{<=>[$\kappa_2$][$\kappa_3$]} 2B
  \arrow(@c1.south east--.north east){0}[-90,.25]
  C \arrow{->[$\kappa_4$]} A
  \schemestop
  \end{split}
\end{equation}
The following holds.
\begin{description}
 \item[The species $A$ is ACR. Moreover, there exists at least one $c>0$ with $g(c)=0$.] Indeed, it can be checked that
 $$g(c)=\begin{pmatrix}
                 -\kappa_1c_1c_2+\kappa_4c_3\\
                 \kappa_2c_2c_3-\kappa_3c_2^2\\
                 \kappa_1c_1c_2-\kappa_2c_2c_3+\kappa_3c_2^2-\kappa_4c_3
                \end{pmatrix}$$
 is zero if and only if $c_2=c_3=0$ or
 $$c=\left(
      \frac{\kappa_3\kappa_4}{\kappa_1\kappa_2},
      s,
      \frac{\kappa_3}{\kappa_2}s
     \right)\quad\text{for some }s\in\RR_{>0}.
$$
%  \item[The boundary equilibria are attracting.] Indeed, if $x\in\RR^4$ is such that $c_2, c_3, c_4<\varepsilon$ for $\varepsilon>0$ small enough, then
 \item[The reaction network has deficiency 1.] Indeed,
 $$S=\spann_\RR\left\{
                    \begin{pmatrix}
                      -1\\ 0\\ 1
                     \end{pmatrix},
                     \begin{pmatrix}
                      0\\ 1\\ -1
                     \end{pmatrix}
\right\}$$
and
 $$\delta=|\C|-\ell-\dim S=5-2-2=1.$$
 \item[There are no non-terminal complexes $y\neq y'$ such that $y'-y$ has only one entry different from 0.] This can be easily checked, since the only non-terminal complexes are $A+B$ and $C$.
 \item[The reaction network is conservative.] Indeed, the vector $(1,1,1)$ is a positive conservation law.
%  \item[The reaction network is bimolecular.] It can be easily checked.
\end{description}

We will now show that no extinction event can occur for the stochastically modeled mass action system, provided that $X_2(0)\geq 1$ and the conserved mass $m=X_1(0)+X_2(0)+X_3(0)\geq 2$. Note that for any time $t\geq0$ we have $X_2(t)\geq1$, since the only reaction decreasing the number of molecules of $B$ is $2B\to B+C$, which removes one molecule of $B$ and can only take place if at least two molecules of $B$ are present. Assume that a state $x'$ is reached from $X(0)$, such that $\lambda_{A+B\to B+C}=0$. Then, it must be that $x'_1=0$. There are two cases:
\begin{itemize}
 \item We have $x_3\geq 1$. Hence, a molecule of $A$ can be created through the occurrence of $C\to A$, and a state is reached where the reaction $A+B\to B+C$ is active.
 \item We have $x_3=0$. Hence, $2\leq m=x'_2$ which means that the reaction $2B\to B+C$ can take place, and this case is reduced to the previous one.
\end{itemize}
In conclusion, no extinction sets for $A+B\to B+C$ are reachable from $X(0)$. It can be shown that the same holds for all other reactions by applying  Proposition~\ref{prop:consecutio} succesively.

\begin{remark}
 By imposing more restricitve assumptions, we can formulate new conjectures. For example, one may be tempted to try to prove that the existence of ACR species implies the occurrence of an extinction event (in the stochastic model) for binary mass action systems in which the coefficient of the species in all complexes are either 0 or 1 (in more biological terms, this implies that there is no autocatalytic production). However, we give here an example showing that this is not true. Consider the mass action system
 
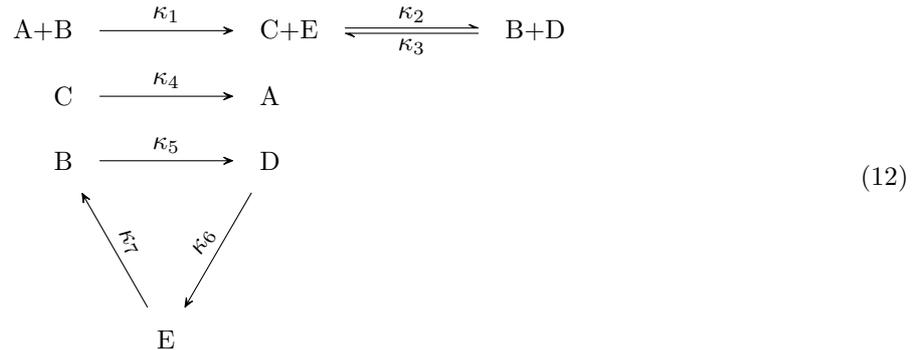
\begin{equation}\label{ex:stoich_1_def_1_cycle}
 \begin{split}
  \schemestart
   A+B\arrow{->[$\kappa_1$]}C+E\arrow{<=>[$\kappa_2$][$\kappa_3$]}B+D
   \arrow(@c1.south east--.north east){0}[-90,.25]
   C\arrow{->[$\kappa_4$]}A
   \arrow(@c4.south east--.north east){0}[-90,.25]
   B\arrow(B--D){->[$\kappa_5$]}D
   \arrow(@D--E){->[$\kappa_6$]}[-120]E
   \arrow(@E--@B){->[$\kappa_7$]}
  \schemestop
 \end{split}
\end{equation}
We  prove in the Appendix that $A$ is an ACR species, and no extinction event can occur for the stochastic model, provided that
\begin{align*}
 X_1(0)+X_2(0)+X_3(0)+X_4(0)+X_5(0)&\geq 2, \\
 X_2(0)+X_4(0)+X_5(0)&\geq 1. 
\end{align*}
Moreover, the only assumption of Theorem~\ref{thm:ACR_stoc} that is not satisfied by \eqref{ex:stoich_1_def_1_cycle} is  \eqref{part:non_term}.
\end{remark}

% \subsection*{Examples 4 and 5: stoichiometric coefficients 1 and deficiency 1}
% 
% Consider
% \begin{equation}\label{ex:stoich_1_def_1_a}
%  \begin{split}
%   \schemestart
%    A+B\arrow{->[1]}B+C\arrow{<=>[1][1]}B+D
%    \arrow(@c1.south east--.north east){0}[-90,.25]
%    C\arrow{->[1]}A
%    \arrow(@c4.south east--.north east){0}[-90,.25]
%    B\arrow{<=>[1][1]}D
%   \schemestop
%  \end{split}
% \end{equation}
% The example is
% \begin{itemize}
%  \item ACR in the species $A$
%  \item conservative
%  \item deficiency 1
%  \item has boundary steady states when the compatibility class does not intersect ACR equilibria
%  \item has all stoichiometric coefficients 1.
% \end{itemize}
% 
% There can be no extinction event, as there is no way to remove completely the molecules of $B$ (or of $D$, which recreates $B$). The nexdt example is similar:
% 

\section{Stationary distributions and equilibria of the associated deterministic model}\label{sec:stable_boundary}

Connections between the equilibria of a deterministic mass action system and the stationary distributions of the corresponding stochastic mass action system have been studied for a special class of models, called complex balanced mass action systems \cite{AC2016,ACK:poisson, CW:poisson, CJ:graphical}. However, in general the existence of positive equilibria for the deterministic mass action system does not imply the positive recurrence of the associated stochastic mass action system, as shown in \eqref{eq:toy_model}, \eqref{ex:endotactic_trans}, and \eqref{ex:endotactic_expl}.

Conversely, it was intuitively thought that the existence of a stationary distribution of the stochastic mass action system would imply the existence of a positive equilibrium of the associated deterministic mass action system. The idea was that a positive equilibrium could be related to the mean of the stationary distribution, or to some sort of weighted average thereof. If that were true, the lack of positive equilibria for the deterministic mass action system would have implied the transience of the states of the associated stochastic model. In particular, for models with a conservative reaction network, the transience of positive states would have implied the absorption at the boundary due to the finiteness of the state space, hence an extinction event.

The fact that lack of positive equilibria in the deterministic mass action system does not imply the transience of the associated stochastic mass action system is, however, shown in \eqref{ex:blow_pos_rec}. However, the question was still open for models with a conservative network, and we close it here with the following bimolecular example.

Consider the bimolecular mass action system
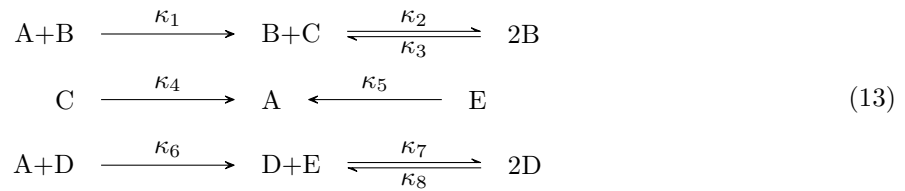
\begin{equation}
 \begin{split}
  \schemestart
  A+B \arrow{->[$\kappa_1$]} B+C \arrow{<=>[$\kappa_2$][$\kappa_3$]} 2B
  \arrow(@c1.south east--.north east){0}[-90,.25]
  C\arrow{->[$\kappa_4$]}A \arrow{<-[$\kappa_5$]} E
  \arrow(@c4.south east--.north east){0}[-90,.25]
  A+D \arrow{->[$\kappa_6$]} D+E \arrow{<=>[$\kappa_7$][$\kappa_8$]} 2D
  \schemestop
 \end{split}
\end{equation}
The following holds.
\begin{description}
 \item[The reaction network is conservative.] Indeed, it can be checked that $(1,1,1,1,1)$ is a positive conservation law.
 \item[There is no positive $c$ with $g(c)=0$ for a general choice of rate constants.] We have
  $$g(c)=\begin{pmatrix}
                 -\kappa_1c_1c_2+\kappa_4c_3+\kappa_5c_5-\kappa_6c_1c_4\\
                 \kappa_2c_2c_3-\kappa_3c_2^2\\
                 \kappa_1c_1c_2-\kappa_2c_2c_3+\kappa_3c_2^2-\kappa_4c_3\\
                 \kappa_7c_4c_5-\kappa_8c_4^2\\
                 -\kappa_5c_5+\kappa_6c_1c_4-\kappa_7c_4c_5+\kappa_8c_4^2
                \end{pmatrix}$$
  By imposing $c>0$, it follows that $g(c)=0$ is equivalent to the system
  $$
  \begin{cases}
   c_3=\frac{\kappa_3}{\kappa_2}c_2\\
   c_3=\frac{\kappa_1}{\kappa_4}c_1c_2\\
   c_5=\frac{\kappa_8}{\kappa_7}c_4\\
   c_5=\frac{\kappa_6}{\kappa_5}c_1c_4\\
  \end{cases}
  $$
  The system has a positive solution if and only if
  $$c_1=\frac{\kappa_3\kappa_4}{\kappa_1\kappa_2}=\frac{\kappa_5\kappa_8}{\kappa_6\kappa_7}.$$
  Hence, no positive equilibria exist if the rate constants do not satisfy the above equality. It is interesting to note that, when a positive equilibrium exists, the above equation implies that the species $A$ is ACR.
  \item[All positive states are positive recurrent. ] We will prove something more: each set of the form
  $$\Upsilon_m=\{x\in\ZZ_{\geq0}^5\,:\, x_2,x_4\geq 1, \|x\|_1=m\}$$
  for some $m\geq 2$ is closed and irreducible. Since the sets $\Upsilon_m$ are finite, it follows that they only contain positive recurrent states. We obtain the desired result by noting that every positive state is contained in some set $\Upsilon_m$, for some $m\geq2$. So, it suffices to show that for a given $m\geq2$ the set $\Upsilon_m$ is closed and irreducible.
  
  We begin by showing that $\Upsilon_m$ is closed. This is equivalent to showing that from a state with a least one molecule of $B$ and one molecule of $D$, it is impossible to reach a state with no molecules of $B$ or no molecules of $D$. This follows from noting that the only reaction decreasing the number of molecules of $B$ is $2B\to B+C$, which decrease the number of molecules of $B$ by one, but it can only occur if at least two molecules of $B$ are present. The same holds for the species $D$, whose molecules can only be decreased through $2D\to D+E$.
  
  To show that $\Upsilon_m$ is irreducible, we can show that for all $x\in\Upsilon_m$, the state $(m-2,1,0,1,0)$ can be reached from $x$, and viceversa. Indeed, if $X(0)=x$ then the reaction $2B\to B+C$ can take place $x_2-1$ times, and the reaction $2D\to D+E$ can take place $x_4-1$ times. Hence, the state $(x_1, 1, x_3+x_2-1, 1, x_5+x_4-1)$ is reached. Now, if all the molecules of $C$ and $E$ are consumed by the reactions $C\to A$ and $E\to A$, then the state $(\|x\|_1-2,1,0,1,0)=(m-2,1,0,1,0)$ is reached. 
  
  Conversely, if $X(0)=(m-2,1,0,1,0)$, then the reaction $A+B\to B+C$ can take place $x_2+x_3-1$ times, and the reaction $A+D\to D+E$ can take place $x_4+x_5-1$ times. The state $(x_1, 1, x_2+x_3-1, 1, x_4+x_5-1)$ is reached. From here, the state $x$ can be reached if the reaction $B+C\to 2B$ takes place $x_2-1$ times and $D+E\to 2E$ takes place $x_5-1$ times.
  \end{description}

\section*{Appendix}
\subsection*{Analysis of the mass action system \eqref{ex:no_def_1_remark}}
Consider the bimolecular mass action system \eqref{ex:no_def_1_remark}, which we repeat here for convenience.
 
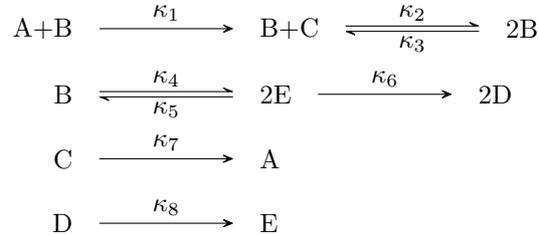
\begin{equation*}
 \begin{split}
  \schemestart
   A+B\arrow{->[$\kappa_1$]}B+C\arrow{<=>[$\kappa_2$][$\kappa_3$]}2B
   \arrow(@c1.south east--.north east){0}[-90,.25]
   B\arrow{<=>[$\kappa_4$][$\kappa_5$]}2E\arrow{->[$\kappa_6$]}2D
   \arrow(@c4.south east--.north east){0}[-90,.25]
   C\arrow{->[$\kappa_7$]}A
   \arrow(@c7.south east--.north east){0}[-90,.25]
   D\arrow{->[$\kappa_8$]}E
  \schemestop
  \end{split}
\end{equation*}
The following holds.
\begin{description}
 \item[The species $A$ is ACR. Moreover, there exists at least one $c>0$ with $g(c)=0$.] Indeed, we have that
 $$g(c)=\begin{pmatrix}
                 -\kappa_1c_1c_2+\kappa_7c_3\\
                 \kappa_2c_2c_3-\kappa_3c_2^2-\kappa_4c_2+\kappa_5c_5^2\\
                 \kappa_1c_1c_2-\kappa_2c_2c_3+\kappa_3c_2^2-\kappa_7c_3\\
                 2\kappa_6c_5^2-\kappa_8c_4\\
                 \kappa_4c_2-2(\kappa_5+\kappa_6)c_5^2+\kappa_8c_4
                \end{pmatrix}$$
 is zero if and only if $c_2=c_3=c_4=c_5=0$ or
 $$c=\left(
      \frac{\kappa_3\kappa_7}{\kappa_1\kappa_2},
      s,
      \frac{\kappa_3}{\kappa_2}s,
      \frac{2\kappa_4\kappa_6}{\kappa_5\kappa_8}s,
      \sqrt{\frac{\kappa_4}{\kappa_5}s}
     \right)\quad\text{for some }s\in\RR_{>0}.
$$
%  \item[The boundary equilibria are attracting.] Indeed, if $x\in\RR^4$ is such that $c_2, c_3, c_4<\varepsilon$ for $\varepsilon>0$ small enough, then
 \item[The reaction network has deficiency 2.] Indeed, 
 $$S=\spann_\RR\left\{
                    \begin{pmatrix}
                      -1\\ 0\\ 1\\ 0 \\ 0
                     \end{pmatrix},
                     \begin{pmatrix}
                      0\\ 1\\ -1\\ 0 \\ 0
                     \end{pmatrix},
                     \begin{pmatrix}
                      0 \\ -1 \\ 0 \\ 0 \\ 2
                     \end{pmatrix},
                     \begin{pmatrix}
                      0\\ 0\\ 0\\ 1 \\ -1
                     \end{pmatrix}
\right\}$$
and
 $$\delta=|\C|-\ell-\dim S=10-4-4=2.$$
 \item[There are two non-terminal complexes $y\neq y'$ such that only the first entry of $y'-y$ is different from 0.] The two complexes $A+B$ and $B$ are non-terminal and their difference is $A$. Note that $A$ is the only ACR species, and $A+B$ and $B$ are the only two non-terminal complexes whose difference has only one entry that is not zero.
 \item[The reaction network is conservative.] Indeed, the vector $(2,2,2,1,1)$ is a positive conservation law.
%  \item[The reaction network is bimolecular.] It can be easily checked.
\end{description}

We will show that no extinction event can take place for the stochastically modeled mass action system, provided that
\begin{align*}
 2X_1(0)+2X_2(0)+2X_3(0)+X_4(0)+X_5(0)&\geq4,\\
 2X_2(0)+X_4(0)+X_5(0)&\geq2.\\
\end{align*}

Assume that a state $x'$ with $\lambda_{A+B\to B+C}(x')=0$ is reachable from $X(0)$. Then, there are two cases.
\begin{itemize}
 \item We have $x'_2=0$. For any $t\geq0$, consider the quantity
 $$h(t)=2X_2(t)+X_4(t)+X_5(t).$$
 The only reaction capable of reducing $h(t)$  is $2B\to B+C$. This reaction decreases $h(t)$ by 2, but it can only take place if at least 2 molecules of $B$ are present, in which case $h(t)\geq4$. Hence, under the assumption that $m(0)\geq 2$, for all $t\geq0$ we necessarily have $h(t)\geq2$. Since $x'$ is reachable from $X(0)$ and $x'_2=0$, we have $x'_4+x'_5\geq2$. By potentially letting the reaction $D\to E$ take place, we may assume that $x'_4\geq2$. Hence, the reaction $2E\to B$ can occur and the number of molecules of $B$ can become positive. At this point, either a state where $A+B\to B+C$ is active is reached, or we consider the following case.
  \item We have $x'_2>0$ and $x'_1=0$.  We have three subcases:
 \begin{itemize}
 \item If $x'_3>0$, then a molecule of $A$ can be created by the occurrence of $C\to A$, which does not modify the number of molecules of $B$. Hence, a state can be reached where $A+B\to B+C$ is active.
 \item  If $x'_3=0$ and $x'_2\geq 2$, then a molecule of $B$ can be transformed into a molecule of $C$ through $2B\to B+C$, which only consumes one molecule of $B$. This subcase is therefore reduced to the previous one.
 \item If $x'_3=0$ and $x'_2=1$, then
 $$4\leq 2x'_1+2x'_2+2x'_3+x'_4+x'_5=0+2+0+x'_4+x'_5,$$
 which implies that $x'_4+x'_5\geq 2$. By potentially using the reaction $D\to E$, we can assume that $x'_4\geq2$ and a molecule of $B$ can be created by the reaction $2E\to B$. This subcase is therefore reduced to the previous one.
 \end{itemize}
\end{itemize}
It follows that no extinction set for $A+B\to B+C$ is reachable from $X(0)$. By consecutive applycations of Proposition~\ref{prop:consecutio}, it follows that the same occurs for all the other reactions as well, so no extinction event can occur.
 
 \subsection*{Analysis of the mass action system \eqref{ex:stoich_1_def_1_cycle}}
 Consider the bimolecular mass action system \eqref{ex:no_def_1_remark}, which we repeat here for convenience.
 
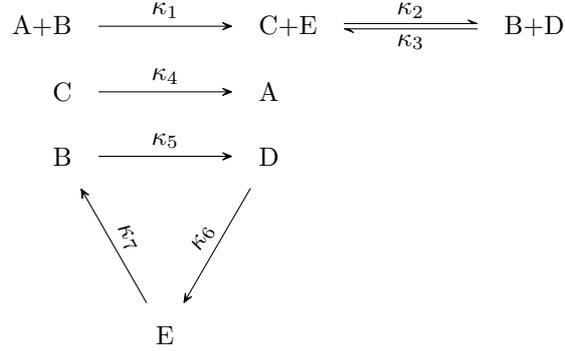
\begin{equation*}
 \begin{split}
  \schemestart
   A+B\arrow{->[$\kappa_1$]}C+E\arrow{<=>[$\kappa_2$][$\kappa_3$]}B+D
   \arrow(@c1.south east--.north east){0}[-90,.25]
   C\arrow{->[$\kappa_4$]}A
   \arrow(@c4.south east--.north east){0}[-90,.25]
   B\arrow(B--D){->[$\kappa_5$]}D
   \arrow(@D--E){->[$\kappa_6$]}[-120]E
   \arrow(@E--@B){->[$\kappa_7$]}
  \schemestop
 \end{split}
\end{equation*}
We have the following.
\begin{description}
 \item[The species $A$ is ACR. Moreover, there exists at least one $c>0$ with $g(c)=0$.] Indeed, it can be checked that
 $$g(c)=\begin{pmatrix}
                 -\kappa_1c_1c_2+\kappa_4c_3\\
                 -\kappa_1c_1c_2+\kappa_2c_3c_5-\kappa_3c_2c_4-\kappa_5c_2+\kappa_7c_5\\
                 \kappa_1c_1c_2-\kappa_2c_3c_5+\kappa_3c_2c_4-\kappa_4c_3\\
                 \kappa_2c_3c_5-\kappa_3c_2c_4+\kappa_5c_2-\kappa_6c_4\\
                 \kappa_1c_1c_2-\kappa_2c_3c_5+\kappa_3c_2c_4+\kappa_6c_4-\kappa_7c_5\\
                \end{pmatrix}$$
 is zero if and only if $c_2=c_3=c_4=c_5=0$ or
 $$c=\left(
      u,
      s,
      \frac{\kappa_1 u}{\kappa_4}s,
      \frac{\kappa_5}{\kappa_6}s,
      \frac{\kappa_5+\kappa_1u}{\kappa_7}s
     \right)\quad\text{for some }s\in\RR_{>0},
$$
where $u$ is the unique positive real number satisfying
$$\kappa_1^2\kappa_2\kappa_6 u^2+\kappa_1\kappa_2\kappa_5\kappa_6 u - \kappa_3\kappa_4\kappa_5\kappa_7=0,$$
namely
$$u=\frac{-\kappa_2\kappa_5\kappa_6+\sqrt{\kappa_2^2\kappa_5^2\kappa_6^2+4\kappa_2\kappa_3\kappa_4\kappa_5\kappa_6\kappa_7}}{2\kappa_1\kappa_2\kappa_6}.$$
%  \item[The boundary equilibria are attracting.] Indeed, if $x\in\RR^4$ is such that $c_2, c_3, c_4<\varepsilon$ for $\varepsilon>0$ small enough, then
 \item[The reaction network has deficiency 1.] Indeed,
 $$S=\spann_\RR\left\{
                    \begin{pmatrix}
                      -1\\ -1\\ 1\\ 0\\1
                     \end{pmatrix},
                     \begin{pmatrix}
                      0\\ 1\\ -1\\1\\-1
                     \end{pmatrix},
                     \begin{pmatrix}
                      1\\ 0\\ -1\\0\\0
                     \end{pmatrix},
                     \begin{pmatrix}
                      0\\ -1\\ 0\\1\\0
                     \end{pmatrix}
\right\}$$
and
 $$\delta=|\C|-\ell-\dim S=8-3-4=1.$$
 \item[There are no non-terminal complexes $y\neq y'$ such that $y'-y$ has only one entry different from 0.] This can be easily checked, since the only non-terminal complexes are $A+B$ and $C$.
 \item[The reaction network is conservative.] Indeed, the vector $(1,1,1,1,1)$ is a positive conservation law.
%  \item[The reaction network is bimolecular.] It can be easily checked.
\end{description}
 We will show that under the assumption
 \begin{align*}
 m=X_1(0)+X_2(0)+X_3(0)+X_4(0)+X_5(0)\geq 2, 
 X_2(0)+X_4(0)+X_5(0)\geq 1,
\end{align*}
 no extinction event can occur for the stochastic model. To this aim, first note that the quantity 
 \[
 	h(t)=X_2(t)+X_4(t)+X_5(t)
\]
 is always greater than or equal to 1. Indeed, the only reaction that can decrease this quantity is 
 \[
 	B+D\to C+E.
\]
However,  under the action of this reaction $h(t)$ decreases by 1, but the reaction is active only if at least one molecule of $B$ and one molecule of $D$ are present, implying $h(t)=2$.
 Assume that a state $x'$ is reachable from $X(0)$, with $\lambda_{A+B\to C+E}(x')=0$. This implies that one of the two following cases occurs.
 \begin{itemize}
  \item We have $x'_2=0$. Since $h(t)\geq1$ for all $t\geq0$, it follows that at least one molecule of $D$ or one molecule of $E$ is present. Hence, a molecule of $B$ can be created by the reactions $D\to E$ and $E\to B$. Either a state where $A+B\to C+E$ is active is reached, or we are in the following case.
  \item We have $x'_1=0$ and $x'_2\geq1$. One of the two following subcases holds.
  \begin{itemize}
   \item We have $x'_3\geq1$. Then, a molecule of $A$ can be created through $C\to A$ and a state is reached where $A+B\to C+E$ is active.
   \item We have $x'_3=0$. Hence,
   $$2\leq m=x'_1+x'_2+x'_3+x'_4+x'_5=x'_2+x'_4+x'_5.$$
   Thanks to the reactions $B\to D$, $D\to E$, and $E\to B$, we can transform all the molecules of $B$, $D$, and $E$ into at least one molecule of $B$ and one molecule of $D$, so that a state where $B+D\to C+E$ is active is reached. Upon the action of   $B+D\to C+E$, a molecule of $C$ is then produced, and this subcase reduces to the previous one.
  \end{itemize}
 \end{itemize}
In conclusion, no extinction set for $A+B\to C+E$ is reachable from $X(0)$. By applying Proposition~\ref{prop:consecutio}, it follows that the same holds for the other reactions. Hence, no extinction event can occur.
 
 \bibliographystyle{plain}
 \bibliography{bib}
\end{document}